\documentclass[a4paper]{article}

\usepackage[margin=1in]{geometry}

\usepackage[english]{babel}
\usepackage{subfig}
\usepackage{graphicx}
\usepackage{amsmath}
\usepackage{amssymb}
\usepackage{xspace}
\usepackage{microtype}
\usepackage{amsmath,amsthm}
\usepackage{hyperref,color}
\usepackage{wrapfig}

\graphicspath{{figures/}}

\newtheorem{observation}{Observation}
\newtheorem{theorem}{Theorem}
\newtheorem{lemma}{Lemma}
\newtheorem{corollary}{Corollary}

\title{Dynamic Graph Coloring\thanks{An extended abstract of this paper appeared in the proceedings of the 15th Algorithms and Data Structures Symposium (WADS 2017)~\cite{BCKLRRV2017Coloring}. M.~K.~was partially supported by MEXT KAKENHI grant Nos.~12H00855, and 17K12635. M.~K., A.~v.~R. and M.~R. were supported by JST ERATO Grant Number JPMJER1201, Japan. L.~B.~was supported by the ETH Postdoctoral Fellowship. S.~V.~was partially supported by NSERC and the Carleton-Fields postdoctoral award. S.~L.~is Directeur de Recherches du F.R.S.-FNRS.}}

\author{Luis Barba\thanks{Dept. of Computer Science, ETH Z\"urich, Switzerland.\texttt{luis.barba@inf.ethz.ch}}
  \and
  Jean Cardinal\thanks{D\'{e}partment d'Informatique, Universit\'{e} Libre de Bruxelles, Brussels, Belgium. \texttt{\{jcardin,stefan.langerman\}@ulb.ac.be}}
  \and
  Matias Korman\thanks{Tohoku University, Sendai, Japan. \texttt{mati@dais.is.tohoku.ac.jp}} \and
        Stefan Langerman\footnotemark[3]
  \and
        Andr\'e van Renssen\thanks{University of Sydney, Sydney, Australia. \texttt{andre.vanrenssen@sydney.edu.au}}
  \and
        Marcel Roeloffzen\thanks{TU Eindhoven, Eindhoven, Netherlands. \texttt{m.j.m.roeloffzen@tue.nl}} 
  \and
        Sander Verdonschot\thanks{School of Computer Science, Carleton University, Ottawa, Canada. \texttt{sander@cg.scs.carleton.ca}}
        }

\newcommand{\N}{\ensuremath{N}\xspace}

\newcommand{\NR}{\ensuremath{N_R}\xspace}
\newcommand{\cR}{\ensuremath{\mathcal{C}}\xspace}
\newcommand{\cRmax}{\ensuremath{\mathcal{C}_{max}}\xspace}

\date{}

\begin{document}
\maketitle

\begin{abstract}
In this paper we study the number of vertex recolorings that an algorithm needs to perform in order to maintain a proper coloring of a graph under insertion and deletion of vertices and edges.
We present two algorithms that achieve different trade-offs between the number of recolorings and the number of colors used.
For any $d>0$, the first algorithm maintains a proper $O(\cR d\N^{1/d})$-coloring while recoloring at most $O(d)$ vertices per update, where $\cR$ and $\N$ are the maximum chromatic number and maximum number of vertices, respectively.
The second algorithm reverses the trade-off, maintaining an $O(\cR d)$-coloring with $O(d\N^{1/d})$ recolorings per update.
The two converge when $d = \log \N$, maintaining an $O(\cR \log \N)$-coloring with $O(\log \N)$ recolorings per update.
We also present a lower bound, showing that any algorithm that maintains a $c$-coloring of a $2$-colorable graph on $\N$ vertices must recolor at least $\Omega(\N^\frac{2}{c(c-1)})$ vertices per update, for any constant $c \geq 2$.
\end{abstract}

\section{Introduction}

It is hard to underestimate the importance of the graph coloring problem in computer science and combinatorics.
The problem is certainly among the most studied questions in those fields, and countless applications and variants have been tackled since it was first posed for the special case of maps in the mid-nineteenth century.
Similarly, the maintenance of some structures in {\em dynamic graphs} has been the subject of study of several volumes in the past couple of decades~\cite{baswana2015fully,baswana2012fully,holm2001poly,roditty2002improved,roditty2004dynamic,thorup2007fully}.
In this setting, an algorithmic graph problem is modelled in the dynamic environment as follows. 
There is an online sequence of insertion and deletion of edges or vertices, and our goal is to maintain the solution of the graph problem after each update. 
A trivial way to maintain this solution is to run the best static algorithm for this problem after each update; however, this is clearly not optimal. A dynamic graph algorithm seeks to maintain some clever data structure for the underlying problem such that the time taken to update the solution is much smaller than that of the best static algorithm. 

In this paper, we study the problem of maintaining a coloring in a dynamic graph undergoing insertions and deletions of both vertices and edges.
At first sight, this may seem to be a hopeless task, since there exist near-linear lower bounds on the competitive factor of online graph coloring algorithms~\cite{HS94}, a restricted case of the dynamic setting.
In order to break through this barrier, we allow a ``fair'' number of {\em vertex recolorings} per update.
We focus on the combinatorial aspect of the problem -- the trade-off between the number of colors used versus the number of recolorings per update.
We present a strong general lower bound and two simple algorithms that provide complementary trade-offs.

\subsubsection{Definitions and Results.}

Let $\cR$ be a positive integer.
A \emph{$\cR$-coloring} of a graph $G$ is a function that assigns a color in $\{1, \ldots, \cR\}$ to each vertex of $G$.
A $\cR$-coloring is \emph{proper} if no two adjacent vertices are assigned the same color.
We say that $G$ is \emph{$\cR$-colorable} if it admits a proper $\cR$-coloring, and we call the smallest such $\cR$ the \emph{chromatic number} of $G$.

A \emph{recoloring algorithm} is an algorithm that maintains a proper coloring of a simple graph while that graph undergoes a sequence of updates. 
Each update adds or removes either an edge or a vertex with a set of incident edges. 
We say that a recoloring algorithm is \emph{$c$-competitive} if it 
uses at most
$c \cdot \cRmax$ colors, where $\cRmax$ is the maximum chromatic number of the graph during the updates.

For example, an algorithm that computes the optimal coloring after every update is $1$-competitive, but may recolor every vertex for every update. 
At the other extreme, we can give each vertex a unique color, resulting in a linear competitive factor for an algorithm that (re)colors at most 1 vertex per update. 
In this paper, we investigate intermediate solutions that use more than $\cR$ colors but recolor a sublinear number of  vertices per update. 
Note that we do not assume that the value $\cR$ is known in advance, or at any point during the algorithm.

In Section~\ref{section:Upper bounds}, we present two complementary recoloring algorithms: an $O(d\N^{1/d})$-competitive algorithm with an amortized $O(d)$ recolorings per update, and an $O(d)$-competitive algorithm with an amortized $O(d\N^{1/d})$ recolorings per update, where $d$ is a positive integer parameter and $\N$ is the maximum number of vertices in the graph during a sequence of updates. 
Interestingly, for $d = \Theta(\log \N)$, both are $O(\log \N)$-competitive with an amortized $O(\log \N)$ vertex recolorings per update. 
Using standard techniques, the algorithms can be made sensitive to the current (instead of the maximum) number of vertices in the graph. 
In addition, we present de-amortized versions of both algorithms in Section~\ref{section:De-amortization}. These de-amortized versions have the same competitive ratio and recolorings per update as the amortized versions. 

We provide lower bounds in Section~\ref{section:Lower bound}. 
In particular, we show that for any recoloring algorithm $A$ using $c$ colors, there exists a specific 2-colorable graph on $\N$ vertices and a sequence of~$m$ edge insertions and deletions that forces $A$ to perform at least $\Omega(m\cdot \N^\frac{2}{c(c-1)})$ vertex recolorings. 
Thus, any $x$-competitive~recoloring algorithm performs in average at least $\Omega(\N^\frac{1}{x(2x-1)})$ recolorings per update.

To allow us to focus on the combinatorial aspects, we assume that we have access to an algorithm that, at any time, can color the current graph (or an induced subgraph) using few colors.
Of course, finding an optimal coloring of an $n$-vertex graph is NP-complete in general~\cite{karp1972reducibility} and even NP-hard to approximate to within $n^{1-\epsilon}$ for any $\epsilon >0$~\cite{zuckerman2007linear}.
Still, this assumption is not as strong as it sounds.
Most practical instances can be colored efficiently~\cite{coudert1997exact}, and for several important classes of graphs the problem is solvable or approximable in polynomial time, including bipartite graphs, planar graphs, $k$-degenerate graphs, and unit disk graphs~\cite{marathe1995simple}.

\subsubsection{Related results.}
\emph{Dynamic graph coloring.} The problem of maintaining a coloring of a graph that evolves over time has been tackled before, but to our knowledge, only from the
points of view of heuristics and experimental results. This includes for instance results from Preuveneers and Berbers~\cite{PB04}, Ouerfelli and Bouziri~\cite{OB11}, 
and Dutot et al.~\cite{DGOP07}. A related problem of maintaining a graph-coloring in an online fashion was studied by Borowiecki and Sidorowicz~\cite{BS12}. In that problem, vertices lose their color, and the algorithm is asked to recolor them. 

\emph{Online graph coloring.} The online version of the problem is closely related to our setting, except that most variants of the online problem only allow the coloring of new vertices, which then cannot be recolored later.
Near-linear lower bounds on the best achievable competitive factor have been proven by Halld\'orsson and Szegedy more than two
decades ago~\cite{HS94}. They show their bound holds even when the model is relaxed to allow a constant fraction of the vertices to change color over the whole sequence. 
This, however, does not contradict
our results. We allow our algorithms to recolor all vertices at some point, but we bound only the number of recolorings {\em per update}. 
Algorithms for online coloring with competitive factor coming close, or equal to this lower bound have been proposed by  Lov{\'{a}}sz et al.~\cite{LST89}, Vishwanathan~\cite{V92},
and Halld{\'{o}}rsson~\cite{H97}.

\emph{Dynamic graphs.} Several techniques have been used for the maintenance of other structures in dynamic graphs, such as spanning trees, transitive closure, 
and shortest paths. Surveys by Demetrescu et al.~\cite{DEGI10,DFI05} give a good overview of those. Recent progress on dynamic connectivity~\cite{DBLP:conf/soda/KapronKM13} and 
approximate single-source shortest paths~\cite{DBLP:conf/soda/HenzingerKN14} are witnesses of the current activity in this field.

\emph{Data structure dynamization.} Our bucketing algorithms are very much inspired by standard techniques for the dynamization of static data structures, pioneered by Bentley and Saxe~\cite{SB79,BS80}, and by Overmars and van Leeuwen~\cite{OL81}.

\subsection{Outline}

In this section, we describe the intuition behind the two complementary recoloring algorithms presented in this paper: the small and the large bucket algorithms.
Both algorithms partition the vertices of the graph into a set of buckets.
Each bucket has $\cR$ colors that are not used by any other bucket. These colors are used to properly color the subgraph induced by the vertices the bucket contains.
This guarantees that the entire graph is always properly colored.
Recall however that we assume that our algorithms have no prior knowledge of the value of $\cR$.

The small-buckets algorithm uses many ``small'' buckets.
This causes it to use more colors, but fewer recolorings per operation.
The buckets are grouped into $d$ levels (for some integer $d>0$), each containing roughly $n^{1/d}$ buckets, and all buckets on the same level have the same capacity (roughly $n^{i/d}$ for level $i$). (Here $n$ denotes the current number of vertices.)
Since each bucket uses at most the first $\cR$ colors of its unique set of colors, the small-buckets algorithm uses a total of $O(d \N^{1/d}\cdot \cR)$ colors.

The idea of the algorithm is simple: every time that an edge is added, we remove one of its endpoints from the bucket it lies in, and we move it to an empty bucket in the first level.
At some point this operation ``fills'' the first level of buckets by filling all buckets at that level. Then all the vertices in this level are promoted to the next level.
This promotion can be propagated again at the next level if that is also filled.
If this propagation reaches the top level, a global recoloring is performed.
Using amortization arguments, we can show that the algorithm performs $O(d)$ amortized recolorings per update.
Intuitively, it suffices to move and recolor a constant number of vertices from each level during each update.

Our second algorithm uses larger buckets and thus uses fewer colors, but more recolorings per operation.
Intuitively, a big bucket is the result of merging all buckets on one level of the small-buckets algorithm.
Thus, we get $d$ buckets, corresponding to the $d$ levels used above, where bucket $i$ has size roughly $n^{(i+1)/d}$.
Since each bucket uses $\cR$ colors, the big-buckets algorithm uses a total of $O(d \cdot \cR)$ colors.
The number of recolorings per insertion however is larger as each insertion triggers a recoloring of the smallest bucket.
Whenever a bucket becomes full, it is emptied into the next bucket, which is in turn recolored.
If this propagation reaches the top level, a global recoloring is performed.
We show that the big-buckets algorithm performs $O(d \N^{1/d})$ amortized recolorings per update.
Using standard de-amortization techniques, we are able to obtain the same bounds in the worst-case for both algorithms.

For the lower bound, consider a graph consisting of three stars with $n/3$ vertices each. 
If a recoloring algorithm wants to maintain a 2-coloring of this graph, then two of the stars will have the same color scheme. 
By linking their roots, we force the algorithm to recolor at least $n/3$ vertices and removing the added edge brings us back to the initial state with the three 2-colored stars, two of them having the same color scheme. Repeating this process shows that any recoloring algorithm that maintains a 2-coloring needs to perform $\Omega(n)$ vertex recolorings per update. 
If we want to maintain a $c$-coloring instead, then this idea can be extended and used in different phases.  
We start constructing trees that are formed by merging stars with the same coloring scheme. 
Our construction builds up larger and larger trees through updates and with every step forces the algorithm to either recolor many vertices or use new colors. 
Eventually the algorithm has used up all its colors and is forced to recolor a large number of vertices.

\section{Upper bound: Recoloring-algorithms}\label{section:Upper bounds}

Before describing the specific strategies, we first introduce some concepts and definitions that are common to all our algorithms.

It is easy to see that deleting a vertex or edge never invalidates the coloring of the graph.
As such, our algorithms do not perform any recolorings when vertices or edges are deleted.
The same is true when an edge is inserted between two vertices of different color, leaving only the insertion of an edge between two vertices of the same color, and the insertion of a new vertex, connected to a given set of current vertices, as interesting cases.
In our algorithms, we simplify this even further, by implementing the edge insertion case as deleting one of its endpoints and re-inserting it with its new set of adjacent edges.
Therefore both the description of the algorithms and the proofs in this section consider only vertex insertions.

Our algorithms partition the vertices into a set of \emph{buckets}, each of which has its own set of colors that it uses to color the vertices it contains.
This set of colors is completely distinct from the sets used by other buckets.
Since all our algorithms guarantee that the subgraph induced by the vertices inside each bucket is properly colored, this implies that the entire graph is properly colored at all times.

The algorithms differ in the number of buckets they use and the size (maximum number of vertices) of each bucket.
Typically, there is a sequence of buckets of increasing size, and one \emph{reset bucket} that can contain arbitrarily many vertices and that holds vertices whose color has not changed for a while.
Initially, the size of each bucket depends on the number of vertices in the input graph.
As vertices are inserted and deleted, the current number of vertices changes.
When certain buckets are full, we \emph{reset} everything, to ensure that we can accommodate the new number of vertices.
This involves emptying all buckets into the reset bucket, computing a proper coloring of the entire graph, and recomputing the sizes of the buckets in terms of the current number of vertices.

We refer to the number of vertices during the most recent reset as $\NR$, and we express the size of the buckets using $s = \lceil \NR^{1/d} \rceil$, where $d>0$ is an integer parameter that allows us to achieve different trade-offs between the number of colors and number of recolorings used.
Since $s = O(\N^{1/d})$, where $\N$ is the maximum number of vertices thus far, we state our bounds in terms of $\N$.
Note that it is also possible to keep $\NR$ within a constant factor of the current number of vertices by triggering a reset whenever the current number of vertices becomes too small or too large.
Standard amortization techniques can be used to show that this would cost only a constant number of additional amortized recolorings per insertion or deletion, although deamortization would be more complicated.
We omit these details for the sake of simplicity.

\subsection{Small-buckets algorithm}

Our first algorithm, called the \emph{small-buckets algorithm}, uses a lot of colors, but needs very few recolorings.
In addition to the reset bucket, the algorithm uses $d s$ buckets, grouped into $d$ \emph{levels} of $s$ buckets each.
All buckets on level $i$, for $0 \leq i < d$, have capacity $s^i$ (see Fig.~\ref{fig:small-buckets}).
Initially, the reset bucket contains all vertices, and all other buckets are empty.
Throughout the execution of the algorithm, we ensure that every level always has at least one empty bucket.
We call this the \emph{space invariant}.

\begin{figure}[ht]
 \centering
 \includegraphics{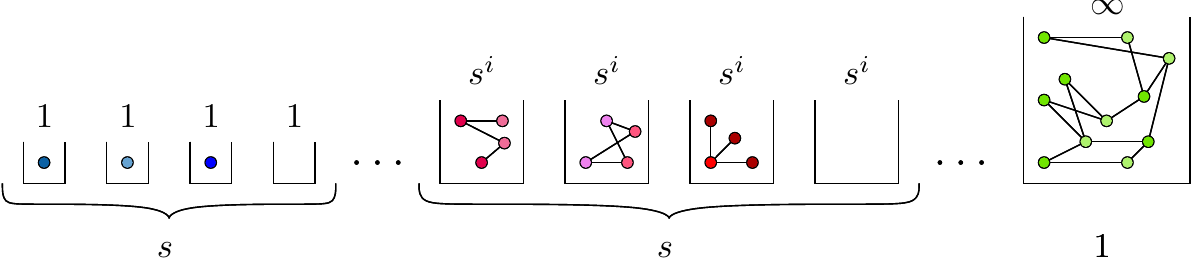}
 \caption{\small The small-buckets algorithm uses $d$ levels, each with $s$ buckets of capacity $s^i$, where $i$ is the level, $s = \lceil \NR^{1/d} \rceil$, and $\NR$ is the number of vertices during the last reset.}
 \label{fig:small-buckets}
\end{figure}

When a new vertex is inserted, we place it in any empty bucket on level $0$.
The space invariant guarantees the existence of this bucket.
Since this bucket has a unique set of colors, assigning one of them to the new vertex establishes a proper coloring.
Of course, if this was the last empty bucket on level $0$, filling it violates the space invariant. 
In that case, we gather up all $s$ vertices on this level, place them in the first empty bucket on level $1$ (which has capacity $s$ and must exist by the space invariant), and compute a new coloring of their induced graph using the set of colors of the new bucket.
If this was the last free bucket on level $1$, we move all its vertices to the next level and repeat this procedure. 
In general, if we filled the last free bucket on level $i$, we gather up all at most $s \cdot s^i = s^{i+1}$ vertices on this level, place them in an empty bucket on level $i + 1$ (which exists by the space invariant), and recolor their induced graph with the new colors.
If we fill up the last level ($d - 1$), we reset the structure, emptying each bucket into the reset bucket and recoloring the whole graph.

\begin{theorem}
\label{thm:ub-small}
For any integer $d>0$, the small-buckets algorithm is an $O(d\N^{1/d})$-competitive recoloring algorithm that uses at most $O(d)$ amortized vertex recolorings per update.
\end{theorem}
\begin{proof}
The total number of colors is bounded by the maximum number of non-empty buckets ($1 + d(s - 1)$), multiplied by the maximum number of colors used by any bucket.
Let $\cR$ be the maximum chromatic number of the graph.
Since any induced subgraph of a $\cR$-colorable graph is also $\cR$-colorable, each bucket requires at most $\cR$ colors.
Thus, the total number of colors is at most $(1 + d(s - 1))\cR$, and the algorithm is $O(d\N^{1/d})$-competitive.

To analyze the number of recolorings, we use a simple charging scheme that places coins in the buckets and pays one coin for each recoloring. 
Whenever we place a vertex in a bucket on level $0$, we give $d + 2$ coins to that bucket.
One of these coins is immediately used to pay for the vertex's new color, leaving $d + 1$ coins.
In general, we maintain the invariant that each non-empty bucket on level~$i$ has $s^i \cdot (d - i + 1)$ coins.

When we merge the vertices on level $i$ into a new bucket on level $i + 1$, we pay a single coin for each vertex that changes color.
Since each bucket had $s^i \cdot (d - i + 1)$ coins, and we recolored at most $s \cdot s^i = s^{i+1}$ vertices, our new bucket has at least $s \cdot s^i \cdot (d - i + 1) - s^{i+1} = s^{i+1} \cdot (d - (i + 1) + 1)$ coins left, satisfying the invariant.

When we fill up level $d - 1$, we reset the structure and recolor all vertices.
At this point, the buckets on level $d - 1$ have a total of $s \cdot s^{d-1} \cdot (d - (d-1) + 1) = 2s^d$ coins, and no more than $s^d$ vertices.
Since all new vertices are inserted on level~$0$, and vertices are moved to the reset bucket only during a reset, the number of vertices in the reset bucket is at most $\NR$.
Since $s^d = \lceil \NR^{1/d} \rceil^d \geq (\NR^{1/d})^d = \NR$, we have enough coins to recolor all vertices.
Thus, we require no more than $d + 2 = O(d)$ amortized recolorings per update.
\end{proof}

\subsection{Big-buckets algorithm}

\begin{wrapfigure}[8]{R}{3in}
\vspace{-2\baselineskip}
 \centering
 \includegraphics{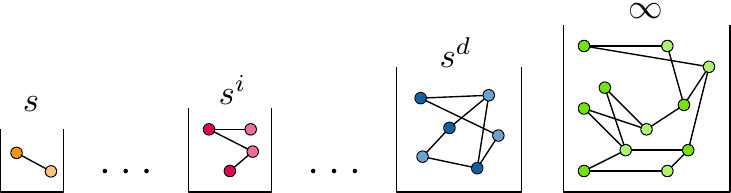}
 \caption{\small Besides the reset bucket, the big-buckets algorithm uses $d$ buckets, each with capacity $s^{i+1}$, where $i$ is the bucket number.}
 \label{fig:big-buckets}
\end{wrapfigure}
Our second algorithm, called the \emph{big-buckets algorithm}, is similar to the small-buckets algorithm, except it merges all buckets on the same level into a single larger bucket.
Specifically, the algorithm uses $d$ buckets in addition to the reset bucket.
These buckets are numbered sequentially from $0$ to $d-1$, with bucket $i$ having capacity $s^{i+1}$, see Fig.~\ref{fig:big-buckets}.
Since we use far fewer buckets, an upper bound on the total number of colors drops significantly, to $(d + 1)\cR$.
Of course, as we will see later, we pay for this in the number recolorings.
Similar to the space invariant in the small-buckets algorithm, the big-buckets algorithm maintains the \emph{high point invariant}: bucket $i$ always contains at most $s^{i+1} - s^i$ vertices (its \emph{high point}).

When a new vertex is inserted, we place it in the first bucket.
Since this bucket may already contain other vertices, we recolor all its vertices, so that the subgraph induced by these vertices remains properly colored.
This revalidates the coloring, but may violate the high point invariant.
If we filled bucket $i$ beyond its high point, we move all its vertices to bucket $i + 1$ and compute a new coloring for this bucket.
We repeat this until the high point invariant is satisfied, or we fill bucket $d - 1$ past its high point.
In the latter case we reset, adding all vertices to the reset bucket and computing a new coloring for the entire graph. 

\begin{theorem}
\label{thm:ub-big}
For any integer $d>0$, the big-buckets algorithm is an $O(d)$-competitive recoloring algorithm that uses at most $O(d\N^{1/d})$ amortized vertex recolorings per update.
\end{theorem}
\begin{proof}
The bound on the number of colors follows directly from the fact that we use $d$ buckets in addition to the reset bucket.
Hence, we use at most $(d+1)\cR$ colors at any point in time, making the algorithms $O(d)$-competitive.
We proceed to analyze the number of recolorings per update.

As in the small-buckets algorithm, we give coins to each bucket that we then use to pay for recolorings. 
In particular, we ensure that bucket~$i$ always has $P_i = \lceil k_i / s^i \rceil \cdot s^{i+1} \cdot (d - i)$ coins, where $k_i$ is the number of vertices in bucket~$i$.

Consider what happens when we place a vertex into bucket $0$.
Initially, the bucket has $P_0 = \lceil k_0 / s^0 \rceil \cdot s^1 \cdot (d- 0) = k_0 s d$ coins. 
As a result of the insertion, we need to recolor all $k_0 + 1$ vertices, and the invariant requires that the bucket has $(k_0 + 1) s d$ coins afterwards.
By the high point invariant, we have that $1 + k_0 \leq s$, so we can bound the number of coins we need to pay per update by $k_0 + 1 + s d \leq (d + 1)s$.

Recall that this insertion may trigger a promotion of all vertices in bucket $0$ to bucket $1$, and that this could propagate until the high point invariant is satisfied again.
When we merge bucket $i$ into bucket $i + 1$, we need to recolor all vertices in these two buckets.
This will be paid for by the coins stored in the smaller bucket.
At this point, the high point invariant gives us that bucket $i$ contains at least $k_i \geq s^{i+1} - s^i + 1$ vertices.
Thus, since $\lceil k_i/s^i \rceil \geq \lceil (s^{i+1} - s^i + 1) / s^i \rceil = s$, bucket $i$ has at least $P_i = \lceil k_i / s^i \rceil \cdot s^{i+1} \cdot (d - i) \geq s^{i+2} \cdot (d-i)$ coins.

At the same time, bucket $i+1$ has $P_{i+1} = \lceil k_{i+1} / s^{i+1} \rceil \cdot s^{i+2} \cdot (d - i - 1)$ coins, and needs to gain at most $s^{i+2} \cdot (d - i - 1)$ coins, as it gains at most $s^{i+1}$ vertices.
This leaves $s^{i+2} \cdot (d - i) - s^{i+2} \cdot (d - i - 1) = s^{i+2}$ coins to pay for the recoloring.
Since bucket $i + 1$ contained no more than $s^{i+2} - s^{i+1}$ vertices by the high-point invariant, and we added at most $s^{i+1}$ new ones, this suffices to recolor all vertices involved and maintain the coin invariant.

Finally, we perform a reset when bucket $d - 1$ passes its high point. 
In that case, bucket $d - 1$ contains at least $s^d - s^{d-1} + 1$ vertices and therefore has at least $\lceil (s^d - s^{d-1} + 1) / s^{d-1} \rceil \cdot s^d \cdot (d - d + 1) = s^{d+1}$ coins.
Since the reset bucket contains at most $\NR \leq s^d$ vertices, we need to recolor at most $2 s^d$ vertices.
As $s = \lceil \NR^{1/d} \rceil \geq 2$ if $\NR \geq 2$, we have enough coins to pay for all these recolorings.
Therefore we can maintain the coloring with $(d + 1)s = O(d\N^{1/d})$ amortized recolorings per update.
\end{proof}

\section{De-amortization}\label{section:De-amortization}

In this section, we show how to de-amortize the algorithms presented in Section~\ref{section:Upper bounds}.
We distinguish between the two different strategies.

\subsection{Shadow vertices}

To de-amortize the two algorithms, we simulate the amortized version using fake vertices, called \emph{shadow vertices}.
Each real vertex $v$ either has a unique shadow vertex $sh(v)$, representing its state (color and location) in the amortized version, or has no shadow vertex, if it would be in the same location and have the same color in the amortized algorithm.
When an update happens, we first move the shadow vertices exactly as the amortized algorithm would (creating new shadow vertices for real vertices without a shadow if needed), and then move and recolor some real vertices to match their shadows, removing the shadows.
We call the first step the \emph{simulation step}, and the second step the \emph{move step}.
Since we only count recolorings of real vertices, only the move step has any actual cost - we just use the simulation step to keep track of where vertices need to go.

The only difference between the simulated versions of the amortized algorithms and the algorithms as presented in Section~\ref{section:Upper bounds} is that the value for $s$ is not allowed to decrease after a reset.
Thus, $s = \lceil \NR^{1/d} \rceil$, where $\NR$ is the \emph{maximum} number of vertices during any reset so far.

\subsection{Small-buckets algorithm}

The de-amortized version of the small-buckets algorithm uses the same buckets as the amortized version, except for an additional reset bucket.
At each stage of the algorithm there is one primary reset bucket and one secondary reset bucket.
The primary reset bucket contains shadow vertices and real vertices without a shadow, whose colors correspond to the reset bucket in the amortized version.
The secondary reset bucket contains real vertices with a shadow in the primary reset bucket.
During a reset, the primary and secondary reset buckets change roles.

As before, we discuss only vertex insertion.
Recall that the amortized algorithm places a new vertex $v$ in an empty bucket on level $0$, and then iteratively merges full levels into higher ones, possibly triggering a reset if the last level fills up.
During the simulation step, the de-amortized algorithm mimics this.
Note that for the purpose of the simulation, we ignore real vertices with a shadow and instead operate on their shadow vertices.
Thus, a level is considered full if every bucket contains either a shadow vertex or a real vertex without a shadow.
On the other hand, whenever the amortized algorithm uses an empty bucket, we require that that bucket contains no real vertices at all, not even ones with a shadow.
We show later that such a bucket is always available when needed.

We first create a new shadow vertex for $v$ and place it in an empty bucket on level $0$.
Then, if this level is full, we create a shadow vertex for every real vertex on this level without one, and move all shadow vertices to an empty bucket on level $1$.
Here, we color them with the new bucket's colors so that their induced graph is properly colored.
If this fills up the new level, we repeat this until we reach a level that is not full, or we fill up the last level.
In the latter case, we trigger a reset.

During a reset, we create a new shadow vertex for all real vertices without one and move all shadow vertices into the secondary reset bucket, computing a proper coloring for them.
At this point, the primary and secondary reset buckets switch roles.
As in the amortized algorithm, we also recompute the value of $s$.
If~$s$ increases, we add additional empty buckets and increase the capacity of the current buckets (recall that we do not allow $s$ to decrease in the de-amortized versions).

All of this happens during the simulation step.
During the move step, we perform the actual recolorings.
We first move and recolor the inserted vertex $v$ to its shadow: moving it into the bucket containing $sh(v)$, giving it the color of $sh(v)$, and removing $sh(v)$.
Then we move and recolor one vertex from each level to its shadow.
Specifically, for each level, we consider all buckets containing only real vertices with a shadow.
Among those buckets, we pick the bucket with the least number of vertices and move and recolor one of its vertices to its shadow.
Finally, we check the secondary reset bucket for vertices with a shadow and move and recolor one if found.
Thus, we recolor at most $d + 2$ vertices per update.
\medskip

\noindent\textbf{Analysis}
We prove correctness by arguing that an empty bucket is available when needed.

\begin{lemma}
\label{lem:0-bucket-empty}
 After every update, there is at least one empty bucket on level $0$.
\end{lemma}
\begin{proof}
 Since the lemma is true for the amortized version by the space invariant, we know that after the simulation step there is at least one bucket on level $0$ that is either empty, in which case we are done, or contains a real vertex with a shadow.
 In this case, the move step will empty one such bucket.
\end{proof}

\begin{lemma}
\label{lem:bound-vertices-in-buckets}
 Let $t_i$ be the number of updates since the last time level $i$ was full, or the last reset, whichever is more recent.
 Then there is a bucket on level $i + 1$ that does not contain any shadow vertices and contains at most $\max(0,s^{i+1} - t_i)$ real vertices, each with a shadow.
\end{lemma}
\begin{proof}
 We prove this by induction on $t_i$.
 The base case $t_i = 0$ follows from the space invariant of the amortized version, since a bucket that is empty in the amortized version will only contain real vertices with a shadow, and each bucket has capacity $s^{i+1}$, which cannot decrease during resets.

 For the inductive step $t_i > 0$, we know that before this update level $i + 1$ had a bucket without shadow vertices that was either completely empty (if $t_i > s^{i+1}$), or had at most $s^{i+1} - (t_i - 1)$ real vertices, each with a shadow.
 If the bucket was already empty, we are done.
 Otherwise, note that during the simulation step, the only time new vertices are created on level $i + 1$ is when level $i$ fills up, which did not happen, as $t_i > 0$.
 During a move step, we move and recolor one vertex from a bucket without shadow vertices and with the least number of real vertices, each with a shadow.
 Therefore, there must now be a bucket without shadow vertices and with at most $s^{i+1} - t_i$ real vertices, each with a shadow.
\end{proof}

\begin{lemma}
\label{lem:fill-takes-long}
 At least $s^{i+1}$ updates are required for level $i$ to fill up again after a reset or after it has filled up.
\end{lemma}
\begin{proof}
 Note that levels only fill up during the simulation step.
 Therefore this is a property of the amortized version of the algorithm.
 We prove the lemma by induction on $i$.
 Since each update creates at most one vertex on level $0$, and there are $s$ buckets, the lemma holds for level $0$.
 Vertices on level $i > 0$ are only created when level $i - 1$ fills up.
 By induction, this happens at most every $s^i$ updates, and each occurrence creates these vertices in only one bucket.
 Since there are $s$ buckets on level $i$, it takes at least $s^{i + 1}$ updates for it to fill up.
\end{proof}

By combining Lemmas~\ref{lem:bound-vertices-in-buckets}~and~\ref{lem:fill-takes-long} we get the following corollary.

\begin{corollary}
\label{cor:mid-buckets-empty}
 When level $i$ fills up, there is an empty bucket on level $i + 1$ (for $0 \leq i < d - 1$).
\end{corollary}

\begin{lemma}
\label{lem:resets-empty}
 When a reset happens, the secondary reset bucket is empty.
\end{lemma}
\begin{proof}
 Initially, the entire point set is contained in the primary reset bucket, and the secondary one is empty.
 Since only a reset can create shadow vertices in a reset bucket, the lemma holds at the first reset.
 When a subsequent reset happens, we have performed at least $s^d \geq \NR$ updates in order to fill up level $d-1$ by Lemma~\ref{lem:fill-takes-long}.
 Since we move one vertex from the secondary reset bucket to the primary reset bucket during each move step, and the secondary reset bucket contains at most $\NR$ vertices, this bucket will be empty when the next reset happens.
\end{proof}

This shows that an empty bucket is available when needed, completing the correctness proof.
 The additional reset bucket increases the number of colors we use by $\cR$ compared to the amortized algorithm, giving the following result.

\begin{theorem}
For any integer $d>0$, the de-amortized small-buckets algorithm is an $O(d\N^{1/d})$-competitive recoloring algorithm that uses at most $d + 2$ vertex recolorings per update.
\end{theorem}

\subsection{Big-buckets algorithm}

As for the small-buckets algorithm, the de-amortized big-buckets algorithm splits the work into a simulation step, in which we move and recolor shadow vertices according to the amortized algorithm, and a move step in which we move and recolor a small number of real vertices to their shadows. The main difference is that we double all buckets, instead of just the reset bucket. There is a primary and secondary version of every bucket, each with its own set of colours. The primary buckets contain shadow vertices and real vertices without a shadow, while the secondary buckets contain only real vertices with a shadow. The simulation step acts only on the primary buckets, while the move step takes real vertices from the secondary buckets to their shadows in the primary buckets. The primary and secondary bucket on a level switch roles when new shadow vertices are added to the level. We prove that this happens only when the secondary bucket is empty.

\captionsetup[subfloat]{justification=centering}
\begin{figure}[htb]
 \centering
 \subfloat[]{\includegraphics{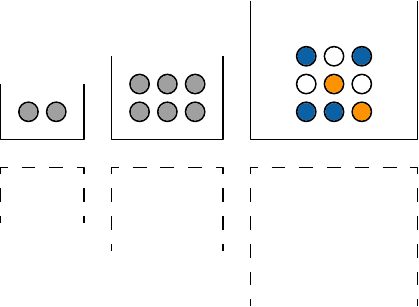}}
 \hspace{0.015\textwidth}
 \subfloat[]{\includegraphics{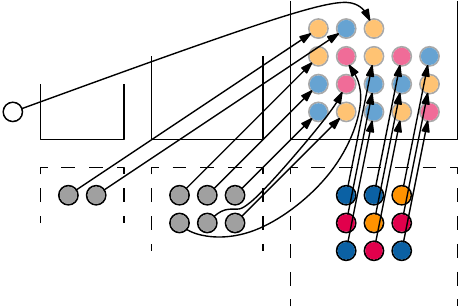}}
 \hspace{0.025\textwidth}
 \subfloat[]{\includegraphics{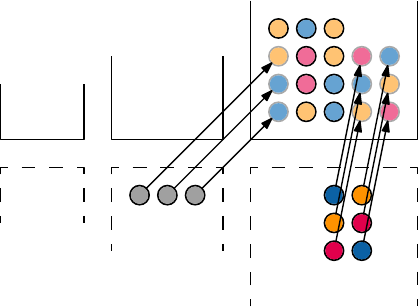}}
 \caption{One vertex insertion in the de-amortized big-buckets algorithm. (a) Before the update. (b) During the simulation step the new vertex causes the first two levels to fill up, creating shadow vertices in the secondary third bucket and swapping the roles of these bucket pairs. (c) During the move step we move and recolor up to $s$ vertices from each level.}
 \label{fig:big-buckets-deamortized-link}
\end{figure}

Recall that the high point invariant states that bucket $i$ contains at most $s^{i+1} - s^i$ vertices. When a vertex is inserted, the amortized big-bucket algorithm tries to place it in bucket $0$. If this violates the high point invariant, the bucket is emptied into the bucket on the next level, and so on, until we reach a level where the high point invariant is not violated. If such a level does not exist, we trigger a reset.

The de-amortized algorithm simulates this as follows. During the simulation step, we find the first level $i$ where we can insert the new vertex and all vertices on lower levels without violating the high point invariant. We give all these vertices, along with the vertices in the primary bucket of level $i$, a shadow in the secondary bucket of level $i$ and compute a coloring for them, see Fig.~\ref{fig:big-buckets-deamortized-link}. We then switch the roles of the primary and secondary buckets for all levels involved. During the move step, we move and recolor $v$ to its shadow. In addition, we move and recolor up to $s$ vertices from each level's secondary bucket and the secondary reset bucket to their shadows.

If we cannot find a level to insert the new vertex without violating the high point invariant, we reset. This involves discarding all current shadow vertices and creating a new shadow vertex in the secondary reset bucket for each real vertex, computing a coloring for these vertices, and switching the primary and secondary reset buckets. We also recompute $s$ and increase the bucket sizes if necessary.
\medskip

\noindent\textbf{Analysis} 
For correctness, the only thing we need to show is that, when we place shadow vertices in a secondary bucket on level $i$, that bucket is empty. A similar argument shows that the secondary reset bucket is empty whenever the high point invariant would fail for level $d - 1$.

\begin{lemma}
 When an update causes us to place shadow vertices in the secondary bucket of level $i$, that bucket is empty.
\end{lemma}
\begin{proof}
 Since the buckets on level $0$ contain fewer than $s$ vertices, the move step after each update empties the secondary bucket. For $i > 0$, recall that we only place shadow vertices here if the high point invariant would have been violated at level $i - 1$, which means that there are at least $s^{i} - s^{i-1} + 1$ real vertices in the levels before $i$. Moreover, none of these vertices have a shadow, since each level's secondary bucket is empty by induction, and the primary bucket contains only real vertices without a shadow.
 Recall that, whenever we place shadow vertices in the secondary bucket of level $i$, we do this for all of the vertices on the lower levels. Thus, these $s^{i} - s^{i-1} + 1$ vertices were inserted after the secondary bucket was last filled. Since each move step moves $s$ vertices from the secondary bucket into the primary bucket, and since the buckets on level $i$ contain no more than $s^{i+1} - s^i$ vertices by the high point invariant, the secondary bucket will be empty before the current insertion.
\end{proof}

The number of colors used is doubled compared to the amortized version, but the number of recolorings per operation is the same: $1$ for the inserted vertex, at most $s - 1$ for level $0$, and at most $s$ for every other level and the reset bucket.

\begin{theorem}
 For any integer $d>0$, the de-amortized big-buckets algorithm is an $O(d)$-competitive recoloring algorithm that uses at most $(d + 1)s = O(dN^{1/d})$ vertex recolorings per update.
\end{theorem}

\section{Lower bound}\label{section:Lower bound}

In this section we prove a lower bound on the amortized number of recolorings for any algorithm that maintains a $c$-coloring of a 2-colorable graph, for any constant $c \geq 2$.
We say that a vertex is \emph{$c$-colored} if it has a color in $[c] = \{1, \ldots, c\}$.
For simplicity of description, we assume that a recoloring algorithm only recolors vertices when an edge is inserted and not when an edge is deleted, as edge deletions do not invalidate the coloring.
This assumption causes no loss of generality, as we can delay the recolorings an algorithm would perform in response to an edge deletion until the next edge insertion.

The proof for the lower bound consists of several parts.
We begin with a specific initial configuration and present a strategy for an adversary that constructs a large configuration with a specific colouring and then repeatedly performs costly operations in this configuration.
In light of this strategy, a recoloring algorithm has a few choices: it can allow the configuration to be built and perform the recolorings required, 
it can destroy the configuration by recoloring parts of it instead of performing the operations, or it can prevent the configuration from being built in the first place by recoloring parts of the building blocks.
We show that all these options require an amortized large number of recolorings.

\subsection{Maintaining a 3-coloring}

To make the general lower bound easier to understand, we first show that to maintain a 3-coloring, we need at least $\Omega(n^{1/3})$ recolorings on average per update. 

\begin{lemma}
\label{lem:3-coloring-lb}
For any sufficiently large $n$ and any $m \geq 2n^{1/3}$, there exists a forest with $n$ vertices, such that for any recoloring algorithm $A$, there exists a sequence of $m$ updates that forces $A$ to perform $\Omega(m \cdot n^{1/3})$ vertex recolorings to maintain a 3-coloring throughout this sequence.
\end{lemma}
\begin{proof}
Let $A$ be any recoloring algorithm that maintains a 3-coloring of a forest under updates.
We use an adversarial strategy to choose a sequence of updates on a specific forest with $n$ nodes that forces $A$ to recolor ``many'' vertices.
We start by describing the initial forest structure.

\begin{figure}[th]
\centering
\includegraphics{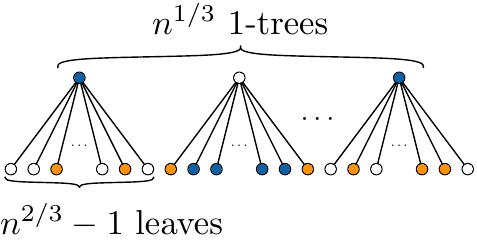} \hspace{3em} \includegraphics{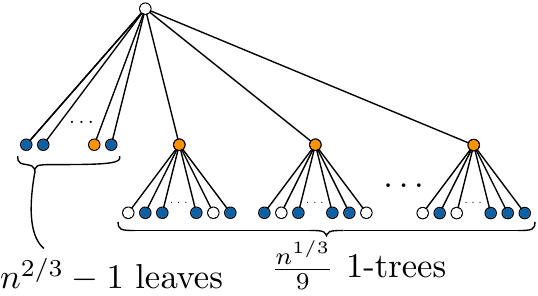}
\caption{\small (left) A 1-configuration is any forest that has many 1-trees as induced subgraphs.
(right) A 2-tree is constructed by connecting the roots of many 1-trees.}
\label{fig:1-configuration}
\end{figure}

A \emph{1-tree} is a rooted (star) tree with a distinguished vertex as its root and $n^{2/3}-1$ leaf nodes attached to it.
Initially, our forest consists of $n^{1/3}$ pairwise disjoint 1-trees, which account for all $n$ vertices in our forest.
The sequence of updates we construct never performs a cut operation among the edges of a 1-tree.
Thus, the forest remains a \emph{1-configuration}: a forest of rooted trees with the $n^{1/3}$ independent 1-trees as induced subgraphs; see Fig.~\ref{fig:1-configuration}~(left).
We require that the induced subtrees are not \emph{upside down}, that is, the root of the 1-tree should be closer to the root of the full tree than its children.
Intuitively, a 1-configuration is simply a collection of our initial 1-trees linked together into larger trees.

Let $F$ be a 1-configuration.
We assume that $A$ has already chosen an initial 3-coloring of $F$.
We assign a color to each 1-tree as follows.
Since each 1-tree is properly 3-colored, the leaves cannot have the same color as the root.
Thus, a 1-tree $T$ always has at least $\frac{n^{2/3}-1}{2}$ leaves of some color $C$, and $C$ is different from the color of the root.
We assign the color $C$ to $T$.
In this way, each 1-tree is assigned one of the three colors.
We say that a 1-tree with assigned color $C$ becomes \emph{invalid} if it has no children of color $C$ left.
Notice that to invalidate a 1-tree, algorithm $A$ needs to recolor at least $\frac{n^{2/3}-1}{2}$ of its leaves.
Since the coloring uses only three colors, there are at least $\frac{n^{1/3}}{3}$ 1-trees with the same assigned color, say $X$.
In the remainder, we focus solely on these 1-trees.

A \emph{2-tree} is a tree obtained by merging $\frac{n^{1/3}}{9}$ 1-trees with assigned color $X$, as follows.
First, we cut the edge connecting the root of each 1-tree to its parent, if it has one.
Next, we pick a distinguished 1-tree with root $r$, and connect the root of each of the other $\frac{n^{1/3}}{9}-1$ 1-trees to $r$.
In this way, we obtain a 2-tree whose root $r$ has $n^{2/3}-1$ leaf children from the 1-tree of $r$, and $\frac{n^{1/3}}{9}-1$ new children that are the roots of other 1-trees; see Fig.~\ref{fig:1-configuration}~(right) for an illustration.
This construction requires $\frac{n^{1/3}}{9}-1$ edge insertions and at most $\frac{n^{1/3}}{9}$ edge deletions (if every 1-tree root had another parent in the 1-configuration).

We build 3 such 2-trees in total. This requires at most $6 (\frac{n^{1/3}}{9}) = \frac{2n^{1/3}}{3}$ updates.
If none of our 1-trees became invalid, then since our construction involves only 1-trees with the same assigned color $X$, no 2-tree can have a root with color $X$.
Further, since the algorithm maintains a 3-coloring, there must be at least two 2-trees whose roots have the same color.
We can now perform a \emph{matching link}, by connecting the roots of these two trees by an edge (in general, we may need to perform a cut first).
To maintain a 3-coloring after a matching link, $A$ must recolor the root of one of the 2-trees and either recolor all its non-leaf children or invalidate a 1-tree.
If no 1-tree has become invalidated, this requires at least $\frac{n^{1/3}}{9}$ recolorings, and we again have two 2-trees whose roots have the same color.
Thus, we can perform another matching link between them.
We keep doing this until we either performed $\frac{n^{1/3}}{6}$ matching links, or a 1-tree is invalidated.

Therefore, after at most $n^{1/3}$ updates ($\frac{2n^{1/3}}{3}$ for the construction of the 2-trees, and $\frac{n^{1/3}}{3}$ for the matching links), we either have an invalid 1-tree, in which case $A$ recolored at least $\frac{n^{2/3} - 1}{2}$ nodes, or we performed $\frac{n^{1/3}}{6}$ matching links, which forced at least $\frac{n^{1/3}}{6} \cdot \frac{n^{1/3}}{9} = \frac{n^{2/3}}{54}$ recolorings.
In either case, we forced $A$ to perform at least $\Omega(n^{2/3})$ vertex recolorings, using at most $n^{1/3}$ updates.

Since no edge of a 1-tree was cut, we still have a valid 1-configuration, where the process can be restarted.
Consequently, for any $m \geq 2n^{1/3}$, there exists a sequence of $m$ updates that starts with a 1-configuration and forces $A$ to perform $\lfloor\frac{m}{n^{1/3}}\rfloor \Omega(n^{2/3}) = \Omega(m\cdot n^{1/3})$ vertex recolorings.
\end{proof}

\subsection{On \texorpdfstring{$k$}{k}-trees}

\begin{wrapfigure}[10]{R}{3.2in}
\vspace{-1\baselineskip}
\centering
\includegraphics{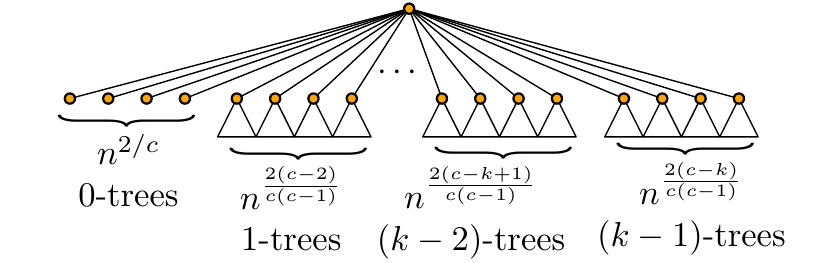}
\caption{\small A $k$-tree is constructed by connecting the roots of a large number of $(k-1)$-trees.}
\label{fig:A c-tree}
\end{wrapfigure}

We are now ready to describe a general lower bound for any number of colors $c$. 
The general approach is the same as when using 3 colors: We construct trees of height up to $c+1$, each excluding a different color for the root of the merged trees. By now connecting two such trees, we force the algorithm $A$ to recolor the desired number of vertices.

A \emph{0-tree} is a single node, and
for each $1\leq k\leq c$, a \emph{$k$-tree} is a tree obtained recursively by merging $2\cdot n^\frac{2(c-k)}{c(c-1)}$  $(k-1)$-trees as follows: 
Pick a $(k-1)$-tree and let $r$ be its root. Then, for each of the $2\cdot n^\frac{2(c-k)}{c(c-1)}-1$ remaining $(k-1)$-trees, connect their root to $r$ with an edge; see Fig.~\ref{fig:A c-tree} for an illustration.

As a result, for each $0\leq j \leq k-1$, a $k$-tree $T$ consists of a root $r$ with $2\cdot n^\frac{2(c-j-1)}{c(c-1)}-1$ $j$-trees, called the \emph{$j$-subtrees} of $T$, whose root hangs from $r$.
The root of a $j$-subtree of $T$ is called a \emph{$j$-child} of $T$.
By construction, $r$ is also the root of a $j$-tree which we call the \emph{core} $j$-tree of $T$.

Whenever a $k$-tree is constructed, it is assigned a color that is present among a ``large'' fraction of its $(k-1)$-children.
Indeed, whenever a $k$-tree is assigned a color $c_k$, we guarantee that it has at least $\left\lceil \frac{2}{c}\cdot n^\frac{2(c-k)}{c(c-1)}\right\rceil$ $(k-1)$-children of color~$c_k$.
We describe later how to choose the color that is assigned to a $k$-tree. 

We say that a $k$-tree that was assigned color $c_k$ has a \emph{color violation} if its root no longer has a $(k-1)$-child with color $c_k$. 
We say that a $k$-tree $T$ becomes \emph{invalid} if either (1) it has a color violation or (2) if a core $j$-tree of $T$ has a color violation for some $1\leq j < k$; otherwise we say that $T$ is \emph{valid}.

\begin{observation}\label{obs:Recolors to invalidate}
To obtain a color violation in a $k$-tree constructed by the above procedure, $A$ needs to recolor at least
$\left\lceil \frac{2}{c}\cdot n^\frac{2(c-k)}{c(c-1)}\right\rceil$ vertices.
\end{observation}
\begin{proof}
Let $T$ be a valid $k$-tree constructed by the above procedure. 
Assume that $T$ was assigned color $c_k$ and recall that by definition, $T$ had at least $\left\lceil \frac{2}{c}\cdot n^\frac{2(c-k)}{c(c-1)}\right\rceil$ $(k-1)$-children of color $c_k$ when its color was assigned. Therefore, in order for $T$ to have a color violation, $A$ needs to change the color of at least $\left\lceil \frac{2}{c}\cdot n^\frac{2(c-k)}{c(c-1)}\right\rceil$ vertices. 
\end{proof}

Notice that a valid $c$-colored $k$-tree of color $c_k$ cannot have a root with color~$c_k$. 
Formally, color $c_k$ is \emph{blocked} for the root of a $k$-tree if this root has a child with color $c_k$. 
In particular, the color assigned to a $k$-tree and the colors assigned to its core $j$-trees for $1 \leq j \leq k-1$ are blocked as long as the tree is valid. 

\subsection{On \texorpdfstring{$k$}{k}-configurations}

A $0$-configuration is a set $F_0$ of $c$-colored nodes, where $|F_0| = T_0 = \alpha n$, for some sufficiently large constant $\alpha$ which will be specified later. For $1\leq k< c$, a \emph{$k$-configuration} is a set $F_k$ of $T_k$ $k$-trees, where
$$T_k = \frac{\alpha}{(4c)^k}\cdot n^{1 - \sum_{i=1}^k \frac{2(c-i)}{c(c-1)}}.$$
Note that the trees of a $k$-configuration may be part of $m$-trees for $m > k$.
If at least $\frac{T_k}{2}$ $k$-trees in a $k$-configuration are valid, then the configuration is \emph{valid}.

For our construction, we let the initial configuration $F_0$ be an arbitrary $c$-colored $0$-configuration in which each vertex is $c$-colored.
To construct a $k$-configuration $F_k$ from a valid $(k-1)$-configuration $F_{k-1}$, consider the at least $\frac{T_{k-1}}{2}$ valid $(k-1)$-trees from $F_{k-1}$. Recall that the trees of $F_{k-1}$ may be part of larger trees, but since we consider edge deletions as ``free'' operations we can separate the trees.
Since each of these trees has a color assigned, among them at least $\frac{T_{k-1}}{2c}$ have the same color assigned to them. Let $c_{k-1}$ denote this color. 

Because each $k$-tree consists of $2\cdot n^\frac{2(c-k)}{c(c-1)}$ $(k-1)$-trees, to obtain $F_k$ we merge $\frac{T_{k-1}}{2c}$ $(k-1)$-trees of color $c_{k-1}$ into $T_k$ $k$-trees, where 
$$T_k = \frac{T_{k-1}}{2c} \cdot \frac{1}{2\cdot n^\frac{2(c-k)}{c(c-1)}} = \frac{\alpha}{(4c)^k}\cdot n^{1 - \sum_{i=1}^k \frac{2(c-i)}{c(c-1)}}.$$

Once the $k$-configuration $F_k$ is constructed, we perform a \emph{color assignment} to each $k$-tree in $F_k$ as follows: 
For a $k$-tree $\tau$ of $F_k$ whose root has \mbox{$2\cdot n^\frac{2(c-k)}{c(c-1)}-1$} $c$-colored $(k-1)$-children, we assign $\tau$ a color that is shared by at least $\left\lfloor \frac{2}{c}\cdot n^\frac{2(c-k)}{c(c-1)}-1\right\rfloor$ of these $(k-1)$-children. 
Therefore, $\tau$ has at least $\left\lfloor \frac{2}{c}\cdot n^\frac{2(c-k)}{c(c-1)}\right\rfloor$ children of its assigned color.
After these color assignments, if each $(k-1)$-tree used is valid, then each of the $T_k$ $k$-trees of $F_k$ is also valid. Thus, $F_k$ is a valid configuration.
Moreover, for $F_k$ to become invalid, $A$ would need to invalidate at least $\frac{T_k}{2}$ of its $k$-trees.

\begin{observation}\label{obs:Property of colored k-trees}
Let $\tau$ be a valid $j$-tree with color $c_j$ assigned to it.
If $r$ is the root of~$\tau$, then $r$ has at least one $(j-1)$-child with color~$c_j$.
\end{observation}

The following result shows how colors are distributed inside a valid $k$-tree.

\begin{lemma}\label{lemma:Subtrees have the same color}
Let $F_k$ be a valid $k$-configuration.
For each $1\leq j < k$, each core $j$-tree of a valid $k$-tree of $F_k$ has color $c_j$ assigned to it. Moreover, $c_i \neq c_j$ for each $1\leq i < j < k$.
\end{lemma}
\begin{proof}
The proof goes by induction on $k$.
For $k= 0$ the results holds trivially.
Assume the result holds for $k-1$.

When constructing $F_{k-1}$ from $F_k$, we know that each $(k-1)$-tree in $F_k$ is assigned the same color $c_{k-1}$. Moreover, by the induction hypothesis, for each $1\leq j < k-1$, each core $j$-tree of a valid $(k-1)$-tree in $F_{k-1}$ had color $c_j$ assigned to it. Thus, each core $j$-tree of a valid $k$-tree also has color~$c_j$ assigned to it. 

We now show that $c_i\neq c_j$ for each $i < j$.
Let $\tau_j$ be a core $j$-tree of a valid $k$-tree in $F_k$ with color $c_j$. 
Since every core $j$-tree of a valid $k$-tree is also valid, $\tau_j$ is a valid $j$-tree. 
Therefore, there is a $(j-1)$-child, say $r$, of $\tau_k$  of color $c_j$. 
Let $\tau_{j-1}$ be the $(j-1)$-subtree of $\tau_j$ rooted at $r$. 
Since $\tau_{j-1}$ has color $c_{j-1}$ assigned to it by the first part of this lemma, we know that its root cannot have color $c_{j-1}$. 
Therefore, $c_j\neq c_{j-1}$ and hence, we can assume that $i< j-1$.

By construction and since $i < j-1$, we know that $r$ is also the root of its core $i$-tree, say $\tau_i$. 
Because $\tau_i$ is valid and has color $c_i$, it must have an $(i-1)$-child $v$ of color $c_i$. 
Since $r$ is the root of $\tau_i$, $r$ and $v$ are adjacent.
Because $r$ has color $c_j$ while $v$ has color $c_i$, and since $F_k$ is $c$-colored, we conclude that $c_i\neq c_j$.
\end{proof}

We also provide bounds on the number of updates needed to construct a $k$-configuration.

\begin{lemma}\label{lemma:Link and cuts to construct k-config}
Using \mbox{$\Theta(\sum_{i=j}^k T_i) = \Theta(T_j)$} edge insertions, we can construct a $k$-configuration from a valid $j$-configuration.
\end{lemma}
\begin{proof}
To merge $\frac{T_{k-1}}{2c}$ $(k-1)$-trees to into $T_k$ $k$-trees, we need $\Theta(T_{k-1})$ edge insertions. Thus, in total, to construct a $k$-configuration from a $j$-configuration, we need $\Theta(\sum_{i=j}^k T_i) = \Theta(T_j)$ edge insertions.
\end{proof}

\subsection{Reset phase}
Throughout the construction of a $k$-configuration, the recoloring-algorithm $A$ may recolor several vertices which could lead to invalid subtrees in $F_j$ for any $1 \leq j < k$. 
Because $A$ may invalidate some trees from $F_j$ while constructing $F_k$ from $F_{k-1}$, one of two things can happen. If $F_j$ is a valid $j$-configuration for each $1\leq j\leq k$, then we continue and try to construct a $(k+1)$-configuration from~$F_k$. Otherwise a \emph{reset} is triggered as follows. 

Let $1\leq j < k$ be an integer such that $F_i$ is a valid $i$-configuration for each $0\leq i\leq j-1$, but $F_j$ is not valid. 
Since $F_j$ was a valid $j$-configuration with at least $T_{j}$ valid $j$-trees when it was first constructed, we know that in the process of constructing $F_k$ from $F_j$, at least $\frac{T_j}{2}$ $j$-trees where invalidated by $A$.
We distinguish two ways in which a tree can be invalid:
\begin{itemize}
\item[] (1) the tree has a color violation, but all its $j-1$-subtrees are valid and no core $i$-tree for $1 \leq i \leq j-1$ has a color violation; or
\item[] (2) A core $i$-tree has a color violation for $1 \leq i \leq j-1$, or the tree has a color violation and at least one of its $(j-1)$-subtrees is invalid.
\end{itemize}
In case (1) the algorithm A has to perform fewer recolorings, but the tree can be made valid again with a color reassignment, whereas in case (2) the $j$-tree has to be rebuild.

Let $Y_0, Y_1$ and $Y_2$ respectively be the set of $j$-trees of $F_j$ that are either valid, or are invalid by case (1) or (2) respectively.
Because at least $\frac{T_j}{2}$ $j$-trees were invalidated, we know that $|Y_1| + |Y_2| > \frac{T_j}{2}$. 
Moreover, for each tree in $Y_1$, $A$ recolored at least  $\frac{2}{c}\cdot n^\frac{2(c-j)}{c(c-1)}-1$ vertices to create the color violation on this $j$-tree by Observation~\ref{obs:Recolors to invalidate}.
For each tree in $Y_2$ however, $A$ created a color violation in some $i$-tree for $i < j$.
Therefore, for each tree in $Y_2$, by Observation~\ref{obs:Recolors to invalidate}, the number of vertices that $A$ recolored is at least 
$\frac{2}{c}\cdot n^\frac{2(c-i)}{c(c-1)}-1 >\frac{2}{c}\cdot n^\frac{2(c-j+1)}{c(c-1)}-1.$

\textbf{Case 1:} $|Y_1| > |Y_2|$. Recall that each $j$-tree in $Y_1$ has only valid $(j-1)$-subtrees by the definition of $Y_1$. 
Therefore, each $j$-tree in $Y_1$ can be made valid again by performing a color assignment on it while performing no update.
In this way, we obtain $|Y_0|+|Y_1|> \frac{T_j}{2}$ valid $j$-trees, i.e., $F_j$ becomes a valid $j$-configuration contained in $F_k$. 
Notice that when a color assignment is performed on a $j$-tree, vertex recolorings previously performed on its $(j-1)$-children cannot be counted again towards invalidating this tree. 

Since we have a valid $j$-configuration instead of a valid $k$-configuration, we ``wasted'' some edge insertions. 
We say that the insertion of each edge in $F_k$ that is not an edge of $F_j$ is a \emph{wasted} edge insertion.
By Lemma~\ref{lemma:Link and cuts to construct k-config}, to construct $F_k$ from $F_j$ we used $\Theta(T_j)$ edge insertions. That is, $\Theta(T_j)$ edge insertions became wasted.
However, while we wasted $\Theta(T_j)$ edge insertions, we also forced $A$ to perform $\Omega(|Y_1|\cdot n^\frac{2(c-j)}{c(c-1)}) = \Omega(T_j \cdot n^\frac{2(c-j)}{c(c-1)})$ vertex recolorings. 
Since $1 \leq j < k \leq c-1$, we know that $n^\frac{2(c-j)}{c(c-1)}  \geq n^\frac{2}{c(c-1)}$. 
Therefore, we can charge $A$ with $\Omega(n^\frac{2}{c(c-1)})$ vertex recolorings per wasted edge insertion.
Finally, we remove each edge corresponding to a wasted edge insertion, i.e., we remove all the edges used to construct $F_k$ from $F_j$. 
Since we assumed that $A$ performs no recoloring on edge deletions, we are left with a valid $j$-configuration $F_j$. 

\textbf{Case 2:} $|Y_2| > |Y_1|$. In this case $|Y_2| > \frac{T_j}{4}$. 
Recall that $F_{j-1}$ is a valid $(j-1)$-configuration by our choice of $j$.
In this case, we say that the insertion of each edge in $F_k$ that is not an edge of $F_{j-1}$ is a \emph{wasted} edge insertion.
By Lemma~\ref{lemma:Link and cuts to construct k-config}, we constructed $F_k$ from $F_{j-1}$ using $\Theta(T_{j-1})$ wasted edge insertions. 
However, while we wasted $\Theta(T_{j-1})$ edge insertions, we also forced $A$ to perform $\Omega(|Y_2| \cdot n^\frac{2(c-j+1)}{c(c-1)}) = \Omega(T_j \cdot n^\frac{2(c-j+1)}{c(c-1)})$ vertex recolorings. 
That is, we can charge $A$ with $\Omega(\frac{T_j}{T_{j-1}}\cdot  n^\frac{2(c-j+1)}{c(c-1)})$ vertex recolorings per wasted edge insertions. 
Since $\frac{T_{j-1}}{T_j} = 4c\cdot n^\frac{2(c-j)}{c(c-1)}$, we conclude that $A$ was charged $\Omega(n^\frac{2}{c(c-1)})$ vertex recolorings per wasted edge insertion.
Finally, we remove each edge corresponding to a wasted edge insertion, i.e., we go back to the valid $(j-1)$-configuration $F_{j-1}$ as before.

Regardless of the case, we know that during a reset consisting of a sequence of $h$ wasted edge insertions, we charged $A$ with the recoloring of $\Omega(h\cdot n^\frac{2}{c(c-1)})$ vertices. 
Notice that each edge insertion is counted as wasted at most once as the edge that it corresponds to is  deleted during the reset phase.
A vertex recoloring may be counted more than once. However, a vertex recoloring on a vertex $v$ can count towards invalidating any of the trees it belongs to. Recall though that $v$ belongs to at most one $i$-tree for each $0\leq i\leq c$. 
Moreover, two things can happen during a reset phase that count the recoloring of $v$ towards the invalidation of a $j$-tree containing it: either (1) a color assignment is performed on this $j$-tree or (2) this $j$-tree is destroyed by removing its edges corresponding to wasted edge insertions. In the former case, we know that $v$ needs to be recolored again in order to contribute to invalidating this $j$-tree. In the latter case, the tree is destroyed and hence, the recoloring of $v$ cannot be counted again towards invalidating it. 
Therefore, the recoloring of a vertex can be counted towards invalidating any $j$-tree at most $c$ times throughout the entire construction. Since $c$ is assumed to be a constant, we obtain the following result. 

\begin{lemma}\label{lemma:Charging recolorings to A}
After a reset phase in which $h$ edge insertions become wasted, we can charge $A$ with $\Omega(h\cdot n^\frac{2}{c(c-1)})$ vertex recolorings. Moreover, $A$ will be charged at most $O(1)$ times for each recoloring.
\end{lemma}

After a reset, we consider our new valid $j$- or $(j-1)$-configuration (depending on the above case), and continue our construction trying to reach a $c$-configuration. 

\subsection{Constructing a \texorpdfstring{$c$}{c}-tree}

If $A$ stops triggering resets, then at some point we reach a $(c-1)$-configuration. 
In this section, we describe what happens when constructing a $c$-configuration from this  $(c-1)$-configuration.
Recall that a color $c_i$ is blocked for the root of a $k$-tree if this root has a child with color $c_i$.

\begin{lemma}\label{lemma:Same colors blocked}
Let $F_k$ be a valid $k$-configuration. 
Then colors $\{c_1, c_2, \ldots, c_{k-1}\}$ are blocked for the root of each valid $k$-tree in $F_k$. 
\end{lemma}
\begin{proof}
Let $\tau$ be a valid $k$-tree in $F_k$ with root $r$. 
Recall that $\tau$ is also the root of a valid $j$-tree for each $1\leq j < k$.
Let $\tau_j$ be the $j$-tree rooted at $r$.
By Lemma~\ref{lemma:Subtrees have the same color}, we know that $\tau_j$ was assigned color $c_j$. By Observation~\ref{obs:Property of colored k-trees}, $r$ has a child of color~$c_j$. Therefore, $r$ has color $c_j$ blocked.
In summary, $r$ has colors $\{c_1, c_2, \ldots, c_{k-1}\}$ blocked.
\end{proof}

A valid $(c-1)$-configuration $F_{c-1}$ consists of at least $\frac{T_{c-1}}{2}$ valid $(c-1)$-trees, where
$$T_{c-1} = \frac{\alpha}{(4c)^{c-1}}\cdot n^{1 - \sum_{i=1}^{c-1} \frac{2(c-i)}{c(c-1)}} = O(1).$$ 
Therefore, by choosing $\alpha$ sufficiently large, we can guarantee that $F_{c-1}$ consists of at least $2(c+1)$ valid $(c-1)$-trees.

Because $F_{c-1}$ is valid, half of its $(c+1)$-trees are valid, i.e., it consists of at least $(c+1)$ valid $(c-1)$-trees.
Because each of these trees has a color assigned to it, among them at least two valid $(c-1)$-trees $\tau$ and $\tau'$ have the same color assigned to them. 
Since $F_{c-1}$ is a valid $(c-1)$-configuration,
Lemma~\ref{lemma:Same colors blocked} implies that each valid $(c-1)$-tree in $F_{c-1}$ has colors $\{c_1, \ldots, c_{c-2}\}$ blocked.
Let $c_{c-1}$ denote the color assigned to $\tau$ and $\tau'$. 

Note that the roots of $\tau$ and $\tau'$ have color $c_{c-1}$ blocked by Observation~\ref{obs:Property of colored k-trees}. 
Moreover, since both $\tau$ and $\tau'$ have colors $\{c_1, \ldots, c_{c-2}\}$ blocked, we conclude that their roots have the same color. 

To construct a $c$-tree, we consider these $2\cdot n^\frac{2(c-c)}{c(c-1)} = 2$ valid $(c-1)$-trees and add an edge connecting their roots.
Since the roots of $\tau$ and $\tau'$ have the same color, $A$ needs to recolor one of them, say $r$. 
However, to recolor $r$ with color $c_i$, it must recolor each child of $r$ with color $c_i$. 
That is, in the core $i$-tree rooted at $r$, $r$ ends with no children of color $c_i$.
Since this $i$-tree has color $c_i$ assigned to it by Lemma~\ref{lemma:Subtrees have the same color}, this makes the core $i$-tree rooted at $r$ invalid and triggers a reset. 
Therefore, every time we reach a $c$-configuration we guarantee that a reset is triggered.

\begin{theorem}\label{thm:lb-general}
Let $c$ be a constant.
For any sufficiently large integers $n$ and $\alpha$ depending only on $c$, and any $m  = \Omega(n)$ sufficiently large, there exists a forest $F$ with $\alpha n$ vertices, such that for any recoloring algorithm $A$, there exists a sequence of $m$ updates that forces $A$ to perform $\Omega(m\cdot n^\frac{2}{c(c-1)})$ vertex recolorings to maintain a $c$-coloring of $F$.
\end{theorem}
\begin{proof}
Use the construction described in this section until $m$ updates have been performed.
Let $m'\leq m$ be the number of edge insertions during this sequence of $m$ updates.
Notice that $m' \geq m/2$ as an edge can only be deleted if it was first inserted and we start with a graph having no edges.

During the construction, $A$ can be charged with $\Omega(n^\frac{2}{c(c-1)})$ vertex recolorings per wasted edge insertion by Lemma~\ref{lemma:Charging recolorings to A}. 
Because the graph in our construction consists of at most $O(n)$ edges at all times, 
at most $O(n)$ of the performed edge insertions are non-wasted.
Since every other edge insertion is wasted during a reset, we know that $A$ recolored $\Omega((m'-n)\cdot n^\frac{2}{c(c-1)})$ vertices. 
Because $m' \geq m/2$ and since $m = \Omega(n)$, our results follows.
\end{proof}

\section{Conclusion}
In this paper we introduced the first method for recoloring few vertices so as to maintain a proper coloring of a large graph with theoretical guarantees. These results give rise to a number of open problems. The obvious one being to close the gap between the upper bounds achieved by our algorithms and the lower bound construction. This question is open even for the case of dynamic forests. It is also worth investigating if a similar lower bound construction can give improved lower bounds for graphs with a higher chromatic number. 

Another thing to note is that our algorithms use the maximum chromatic number. This is undesirable when the graph starts with high chromatic number, but after a number of operations has far lower chromatic number, for example because a number of edges of a large clique are deleted. In this case, an upper bound on the number of colors and recolorings that uses the current chromatic number instead of the maximum would be better. 

Finally, there are a number of different models to consider. For example, can we improve the algorithms when we support only a subset of the operations, such as only vertex insertion? A number of operations can be simulated using the other operations (for example, vertex removal can be effectively achieved by removing all edges to the vertex and ignoring it in the rest of the execution), however this changes the number of operations we execute, possibly allowing fewer recolorings per operation. Similarly, it is interesting to see what happens when we allow different operations, such as edge flips on triangulations or edge slides for trees?

\bibliographystyle{abbrv} 
\bibliography{dynamicColoring}

\end{document}